\documentclass[10pt,twocolumn,twoside]{IEEEtran} 
\usepackage{url}
\usepackage{mathrsfs}

\usepackage{graphics} 
\usepackage{epsfig} 
\usepackage{mathptmx} 
\usepackage{times} 
\usepackage{amsmath} 
\usepackage{amsthm}   
\usepackage{amssymb}  
\usepackage{graphicx}
\usepackage{epstopdf}
\usepackage{bm}
\usepackage{makecell}
\usepackage{booktabs}
\usepackage{multirow}
\usepackage{threeparttable}
\usepackage{enumerate}
\usepackage{subfig}
\usepackage[noadjust]{cite}
\usepackage{multirow}
\usepackage[english]{babel}
\usepackage[autostyle]{csquotes}
\usepackage{xcolor}
\usepackage{slashbox}
\usepackage{lipsum}

\newtheorem{Remark}{Remark}

\newtheorem{Lemma}{Lemma}
\newtheorem{Theorem}{Theorem}
\newtheorem{Corollary}{Corollary}
\begin{document}
%
\title{Global Synchronization of Pulse-Coupled Oscillator Networks Under Byzantine Attacks}
%
%
%

\author{Zhenqian~Wang,~\IEEEmembership{Member,~IEEE,} and
	Yongqiang~Wang,~\IEEEmembership{Senior Member,~IEEE}
	\thanks{Zhenqian Wang and Yongqiang Wang are with the Department of Electrical and Computer Engineering, Clemson University, Clemson, SC, 29630 USA. e-mail:~\{zhenqiw,yongqiw\}@clemson.edu. The work was supported in part by the National Science Foundation under Grants 1738902, 1912702, and in part by the China Scholarship Council.}}

%
%

\markboth{}%
{}
\maketitle
\newcommand\blfootnote[1]{%
\begingroup
\renewcommand\thefootnote{}\footnote{#1}%
\addtocounter{footnote}{-1}%
\endgroup}
\begin{abstract}
	
Synchronization of pulse-coupled oscillators (PCOs) has gained significant attention recently due to their increased applications in sensor networks and wireless communications. Given the distributed and unattended nature of wireless sensor networks, it is imperative to enhance the resilience of PCO synchronization against malicious attacks. However, most existing results on attack-resilient pulse-based synchronization are obtained under assumptions of all-to-all coupling topologies or restricted initial phase distributions. In this paper, we propose a new pulse-based synchronization mechanism to improve the attack resilience of PCO synchronization that is applicable to non-all-to-all networks. Under the proposed synchronization mechanism, we prove that perfect synchronization of legitimate oscillators can be guaranteed in the presence of multiple Byzantine attackers who can emit attack pulses arbitrarily without any constraint except that practical bit rate constraint renders the number of pulses from an attacker to be finite. The new mechanism can guarantee synchronization even when the initial phases of all legitimate oscillators are arbitrarily distributed in the entire oscillation period, which is in distinct difference from most existing attack-resilient synchronization approaches (including the seminal paper from Lamport and Melliar-Smith \cite{lamport1985synchronizing}) that require a priori (almost) synchronization among legitimate oscillators. Numerical simulation results are given to confirm the theoretical results.
	
\end{abstract}
\begin{IEEEkeywords}
Pulse-Coupled Oscillators, Global Synchronization, Byzantine Attacks.
\end{IEEEkeywords}

\section{Introduction}

Inspired by flashing fireflies and contracting cardiac cells, pulse-based synchronization is attracting increased attention in sensor networks and wireless communications \cite{mirollo1990synchronization,peskin1975mathematical,mathar1996pulse,simeone2008distributed}. By exchanging simple and identical messages (so-called pulses), pulse-based synchronization incurs much less energy consumption and communication overhead compared with conventional packet-based synchronization approaches \cite{pagliari2011scalable}. These inherent advantages make pulse-based synchronization extremely appealing for event coordination and clock synchronization in various networks \cite{werner2005firefly,hong2005scalable,hu2006scalability,leidenfrost2009firefly,nunez2017pulse}. In the past decade, plenty of results have been reported on pulse-based synchronization. For example, by optimizing the interaction function, i.e., phase response function, the synchronization speed of pulse-coupled oscillators (PCOs) is maximized in \cite{wang_tsp2:12}; with a judiciously-added refractory period in the phase response function, the energy consumption of PCO synchronization is reduced in \cite{konishi:08,okuda2011experimental,wang_tsp:12}; \cite{wang_tsp:13,nunez2015synchronization,nunez2016synchronization} show that PCOs can achieve synchronization under a general coupling topology even when their initial phases are randomly distributed in the entire oscillation period. Recently, synchronization of PCOs in the presence of time-delays and unreliable links is also discussed \cite{klinglmayr2012guaranteeing,klinglmayr2017convergence}. Other relevant results include \cite{canavier2010pulse,nishimura2011robust,nishimura2012probabilistic,lucken2012two,nunez2015global,proskurnikov2016synchronization,kannapan2016synchronization,lyu2018global,gao2018pulse}.


However, all the above results are obtained under the assumption that all oscillators behave correctly with no nodes compromised by malicious attackers. Due to the distributed and unattended nature, wireless sensor nodes are extremely vulnerable to attacks, making it imperative to study synchronization in the presence of attacks. Recently, some results have emerged for attack-resilient pulse synchronization \cite{tyrrell2010does,klinglmayr2012self,yun2015robustness,wang2018attack,wang2018pulse,wang2020attack,khanchandani2016self,dolev2004self,daliot2003linear,ben2008fast,daliot2003self,dolev2007byzantine,leidenfrost2010fault}. In \cite{tyrrell2010does}, the authors showed that pulse-based synchronization is more robust than its packet-based counterpart in the presence of a faulty node. In \cite{klinglmayr2012self}, a new phase response function was proposed to combat non-persistent random attacks. The authors in \cite{yun2015robustness} considered pulse-based synchronization in the presence of faulty nodes which fire periodically ignoring neighboring nodes' influence. However, all the above results only apply to a priori synchronized PCO network, i.e., all legitimate nodes are required to have identical phases when faulty pulses are emitted. Furthermore, these results also require that the communication topology of legitimate oscillators is all-to-all. 

To relax the constraint on initial phase distributions, \cite{wang2018attack} proposed a pulse-based synchronization approach that is applicable even when legitimate oscillators have different but restricted initial phases; \cite{wang2018pulse} further proposed a pulse-based synchronization mechanism that can achieve synchronization under stealthy attacks even when the initial phases of legitimate oscillators are randomly distributed in the entire oscillation period (global synchronization) under all-to-all connection. The authors in \cite{daliot2003linear,dolev2004self,ben2008fast,khanchandani2016self} proposed to achieve global synchronization by exchanging packets besides pulses. 

On the other hand, to relax the constraint on all-to-all topology, our most resent result \cite{wang2020attack} proposed a new attack resilient pulse-based interaction mechanism to synchronize non-all-to-all connected PCOs when their initial phases are restricted in a half cycle; the authors in \cite{daliot2003self,dolev2007byzantine} employ extra packet based communication to achieve global synchronization under Byzantine attacks even when the network is generally connected. Using a similar approach, the authors in \cite{leidenfrost2010fault} showed that a $(5f+1)$-connected network can achieve global synchronization in the presence of $f$ attackers with each attacker unable to send two attack pulses in one natural oscillation period. Because of the introduction of extra packet messages, these approaches have higher communication and computation overhead, which will further restrict scalability as well as achievable synchronization accuracy.  

\begin{table*}[htbp]\label{Table_Global}
\centering
\centering{Table $1$. Comparison of attack-resilient pulse synchronization approaches.}
\begin{center}
\begin{threeparttable}	
\renewcommand\arraystretch{2}
\scalebox{1}{\begin{tabular}{|c |c |c | c| c|}\hline
\multirow{2}{*}{\makecell{\backslashbox{Approaches}{Merits}}}& \multirow{2}{*}{\makecell{Unrestricted phase\\ distribution conditions}} & \multirow{2}{*}{\makecell{Not restricted to \\ all-to-all networks}} & \multirow{2}{*}{\makecell{Attack model is \\ Byzantine attacks}} & \multirow{2}{*}{\makecell{Communication \\ uses content-free\\ pulses only}}\\ 

{}& {} & {}& {} &{} \\ \hline 
					
{\makecell{\cite{klinglmayr2012self,yun2015robustness,wang2018attack}}}& {$\times$} & {$\times$} & {$\times$} & {$\surd$}\\ \hline

\makecell{\cite{wang2018pulse}}& {$\surd$} & {$\times$} & {$\times$} & {$\surd$}\\ \hline		
\makecell{\cite{wang2020attack}}& {$\times$} & {$\surd$} & {$\times$} & {$\surd$}\\ 
\hline	
\makecell{\cite{khanchandani2016self,dolev2004self,daliot2003linear,ben2008fast}}& {$\surd$} & {$\times$} & {$\surd$} & {$\times$}\\
\hline							
\makecell{\cite{daliot2003self} \cite{dolev2007byzantine}}& {$\surd$} & {$\surd$} & {$\surd$} & {$\times$}\\ 
\hline		
\makecell{\cite{leidenfrost2010fault}}& {$\surd$} & {$\surd$} & {$\times$} & {$\times$}\\ \hline			
\makecell{This paper}& {$\surd$} & {$\surd$} & {$\surd$} & {$\surd$}\\ 
\hline	
\end{tabular}}
\end{threeparttable}
\end{center}
\end{table*}

In this paper, we propose an approach to synchronizing densely connected PCO networks from an arbitrary initial phase distribution under Byzantine (arbitrary) attacks. The approach only employs content-free pulses. It is worth noting that the content-free pulse-based communication reduces the attack surface and avoids the manipulation of message contents by Byzantine attacks. In fact, what can be manipulated by Byzantine attacks becomes the timing of attack pulses, which will be elaborated in Sec. \uppercase\expandafter{\romannumeral3}.

Table $1$ summarizes the advantage of our approach over existing results on pulse-based synchronization. More specifically, compared with existing results, our contributions are as follows: 1) Under Byzantine attacks, our proposed mechanism can synchronize legitimate oscillators even when their initial phases are arbitrarily distributed in the entire oscillation period; 2) Our mechanism is applicable to densely connected PCO networks that are not necessarily all-to-all; 3) We consider an attack model that is much more difficult to deal with than existing results like \cite{klinglmayr2012self,yun2015robustness,wang2018pulse,wang2018attack,wang2020attack}; 4) Our mechanism only use contend-free pulses, which is different from \cite{daliot2003linear,dolev2004self,ben2008fast,khanchandani2016self,daliot2003self,dolev2007byzantine,leidenfrost2010fault} relying on the assistance of packet communication to achieve synchronization; 5) Our proposed mechanism guarantees that the collective oscillation period is identical to the free-running period irrespective of attacks, which is superior to existing mechanisms (e.g., \cite{yun2015robustness,wang2018pulse,wang2018attack}) that lead to a collective oscillation period affected by attacker pulses.

It is worth noting that the results in this paper are fundamentally different from our recent result \cite{wang2020attack} in the following aspects: $1)$ The attack model in this paper is much stronger. \cite{wang2020attack} considers an attack model in which an attacker is restricted to send at most one attack pulse in any time interval of length $T/2$ (to stay stealthy) whereas this paper allows attackers to send as many attack pulses as possible under a given communication channel with a fixed bit rate. So synchronization under attacks in this paper is much more challenging;  $2)$ This paper has more relaxed requirement on the initial distribution of oscillator phases compared with \cite{wang2020attack}. \cite{wang2020attack} requires legitimate oscillators to have initial phases contained in a half cycle whereas this paper allows legitimate oscillators' phases to be arbitrarily distributed in the entire cycle; $3)$ This paper proves finite-time synchronization whereas \cite{wang2020attack} only proves asymptotic synchronization even in the case of $l=1$. More specifically, \cite{wang2020attack} proves that the length of the containing arc of legitimate oscillators will decrease to no greater than $(1-l/2)$ of its original value after every two consecutive firing rounds, and hence can only yield synchronization when time goes to infinity. (It is worth noting that our prior result on non-all-to-all PCO networks in \cite{wang2018attack} needs $0<l<1$ to address the practical case of non-identical initial phases of legitimate oscillators and hence also only proves asymptotic synchronization.)

This paper is organized as follows. Sec. \uppercase\expandafter{\romannumeral2} reviews the main concepts of PCO networks. Sec. \uppercase\expandafter{\romannumeral3} presents the attack model considered in this paper. Sec. \uppercase\expandafter{\romannumeral4} presents a new pulse-based synchronization mechanism. Sec. \uppercase\expandafter{\romannumeral5} addresses the case of multiple Byzantine attackers and Sec. \uppercase\expandafter{\romannumeral6} addresses the case where the total number of oscillators is unknown to individual oscillators. Simulation results are presented in Sec. \uppercase\expandafter{\romannumeral7}.

\section{Preliminaries}
Consider a network of $N$ pulse-coupled oscillators. Each oscillator is equipped with a phase variable which evolves clockwise on a unit circle. When the evolving phase of an oscillator reaches $2\pi$ rad, the oscillator fires (emits a pulse). Receiving pulses from neighboring oscillators will lead to the adjustment of the receiving oscillator's phase, which can be designed to achieve a desired collective behavior such as phase synchronization. To define synchronization, we first introduce the concept of containing arc. The containing arc of legitimate oscillators is defined as the shortest arc on the unit circle that contains all legitimate oscillators' phases.

\emph{Definition 1 (Phase Synchronization)}: We define phase synchronization as a state on which all legitimate oscillators have identical phases and fire simultaneously with a period of $T=2\pi$ seconds.

An edge $(i,j)$ from oscillator $i$ to oscillator $j$ means that oscillator $j$ can receive pulses from oscillator $i$ but not necessarily vice versa. The number of edges entering oscillator $i$ is called the indegree of oscillator $i$ and is represented as $d^-_i$. The number of edges leaving oscillator $i$ is called the outdegree of oscillator $i$ and is represented as $d^+_i$. The value $d_i \triangleq\min\{d^-_i,d^+_i\}$ is called the degree of oscillator $i$. The degree of a network is defined as $d\triangleq\min_{i=1,2,\cdots,N}\{d_i\}$. Since an oscillator cannot receive the pulse emitted by itself, the maximal degree of a network of $N$ PCOs is $d=N-1$, meaning that the network is all-to-all connected. In this paper, we consider dense networks where the network degree $d$ is assumed to be greater than $\lfloor 2N/3 \rfloor$. Making use of the fact $d\triangleq\min_{i=1,2,\cdots,N}\{d_i\}$, we always have $d_i-\lfloor 2N/3\rfloor-1\geq 0$ for $i=1,\,2,\,\cdots,N$.

\section{Attacker Model}

In this section, we present the model of Byzantine attacks. We assume that Byzantine attacks are able to compromise an oscillator and completely take over its behavior. Since the communicated messages in PCO networks are identical and content-free, i.e., pulses, a Byzantine attacker cannot manipulate the content of pulses, but rather, it will judiciously craft attacks via injecting pulse trains at certain time instants to negatively affect pulse-based synchronization.

Because in realistic wireless sensor networks (WSNs), the bit rate of a communication channel between two connected oscillators is limited, an attacker cannot send infinitely many pulses in any finite time interval. In other words, there is always a nonzero time interval between two consecutive pulses from an attacker. Therefore, Byzantine attackers will launch attacks with a time separation greater than $\epsilon$ seconds, where $\epsilon$ is the minimum time separation between two consecutive pulses that can be conveyed by a channel. We summarize the Byzantine attacker model in this paper as follows: 

\emph{Byzantine Attacker}: a Byzantine attacker will emit attack pulses with a time separation greater than $\epsilon$ seconds, where $\epsilon$ is the minimum time separation between two pulses that can be successfully conveyed by a communication channel.

\begin{Remark}\label{Remark_4_1_Global}
In PCO networks, the communication messages are all content-free pulses. So the transmission of one pulse will only occupy the communication channel for a very short time. Only after finishing transmitting one pulse, an attacker can initiate the transmission of another attack pulse. Hence, $\epsilon$ is determined by the length of the pulse and the bit rate of the communication channel. For example, the bit rate of the IEEE 802.15.4 channel is $250kbps$. If we use a control packet (21 bytes) to realize a pulse, then transmitting such pulses will need time separation $\epsilon=(21\times8)/250000=0.672\times10^{-3}$ seconds \cite{leidenfrost2009firefly,zong2018synchronization}.
\end{Remark}

\begin{Remark}\label{Remark_4_2_Global_1}
All existing attack patterns considered under pulse-based synchronization such as random attacks \cite{tyrrell2010does,klinglmayr2012self}, static attacks  \cite{yun2015robustness}, and stealthy attacks \cite{wang2018attack,wang2018pulse,wang2020attack} are special cases of the attacker model considered in this paper.
\end{Remark}

\section{A New Pulse-Based Synchronization Mechanism}

Motivated by the fact that the conventional pulse-based synchronization mechanism is vulnerable to attacks, we propose a new pulse-based synchronization mechanism to combat attacks. To present our new mechanism, we first describe the conventional pulse-based synchronization mechanism.

\noindent\rule{9cm}{0.12em}
\emph{Conventional Pulse-Based Synchronization Mechanism \cite{yun2015robustness}:} \\
\rule{9cm}{0.1em}
\begin{enumerate}
	\item The phase $\phi_i$ of oscillator $i$ evolves from $0$ to $2\pi$ rad with a constant speed $\omega=1$ rad/second.
	
	\item Once $\phi_i$ reaches $2\pi$ rad, oscillator $i$ fires and resets its phase to $0$.
	
	\item \emph{Whenever} oscillator $i$ receives a pulse, it \emph{instantaneously} resets its phase to:
	\begin{equation}\label{PhaseJump}
	\phi_i^+=\phi_i+l\times F(\phi_i)\\
	\end{equation}
	where $l\in(0,1]$ is the coupling strength and $F(\bullet)$ is the phase response function (PRF) given below:
	\begin{equation}\label{PRF}
	\begin{array}{ccc}
	F(\phi):=\left\{
	\begin{array}{ll}
	\hspace{0.1cm} -\phi~~~~~~~~~0\leq\phi\leq\pi \\
	\hspace{0.1cm} 2\pi-\phi~~~~~\pi<\phi\leq 2\pi 
	\end{array}\right.
	\end{array}
	\end{equation}
	For $l=1$, oscillator $i$ will fire immediately if it has $\phi_i^+=2\pi$ rad.
\end{enumerate}
\rule{9cm}{0.12em}

In the above conventional pulse-based synchronization mechanism, every incoming pulse will trigger a jump on the receiving oscillator's phase, which makes it easy for attackers to perturb the phases of legitimate oscillators and hence destroy their synchronization. Moreover, we have that synchronization can never be maintained when attackers only affect part of the network, even when the coupling strength is set to $l=1$. This is because attack pulses can always exert nonzero phase shifts on affected legitimate oscillators and make them deviate from unaffected ones. This is also confirmed by numerical simulation results in Figure \ref{Attack_N_known_Compare} and Figure \ref{Attack_N_unknown_Compare}, which illustrate that existing results in \cite{yun2015robustness,wang2018attack,wang2018pulse} cannot achieve synchronization in the presence of Byzantine attacks when the topology is not all-to-all.

To overcome the inherent vulnerability of existing pulse-based synchronization approaches, we propose a new pulse-based synchronization mechanism (Mechanism $1$) to improve the attack resilience of PCO networks. Our key idea to enable attack resilience is a ``pulse response mechanism" which can restrict the number of pulses able to affect a receiving legitimate oscillator's phase in any oscillation period and a ``phase resetting mechanism" which resets the phase value of a legitimate oscillator upon reaching phase $2\pi$ rad to different values depending on the number of received pulses. The ``pulse response mechanism" and the ``phase resetting mechanism" only allow pulses meeting certain conditions to affect a receiving oscillator's phase and hence can effectively filter out attack pulses with extremely negative effects on the synchronization process. Noting that all pulses are identical and content-free, Mechanism $1$ is judiciously designed based on the number of pulses an oscillator received in the past, i.e., based on memory. The new pulse-based synchronization mechanism is detailed below:

\noindent\rule{9cm}{0.12em}
\emph{New Pulse-Based Synchronization Mechanism (Mechanism 1):}\\
\rule{9cm}{0.1em}
\begin{enumerate}
	\item The phase $\phi_i$ of legitimate oscillator $i$ evolves from $0$ to $2\pi$ rad with a constant speed $\omega=1$ rad/second.
	
	\item Once $\phi_i$ reaches $2\pi$ rad at time $t$, oscillator $i$ fires (emits a pulse) if it did not fire within $(t-\epsilon,\,t]$ and an entire period $T=2\pi$ seconds has elapsed since initiation. Then oscillator $i$ resets its phase from $2\pi$ rad to $0$ if it received over $\lfloor N/3 \rfloor$ pulses within $(t-\epsilon,\,t]$, where $\lfloor \bullet \rfloor$ is the largest integer no greater than $``\bullet."$ Otherwise, it resets its phase from $2\pi$ rad to $\pi$ rad.
	
	\item When oscillator $i$ receives a pulse at time $t'$, it shifts its phase to $2\pi$ rad only if $\phi_i\in[\pi,\,2\pi]$ at time instant $t'$ and one of the following conditions is satisfied:
	\begin{enumerate}		
	\item before receiving the current pulse, oscillator $i$ has already received at least $d_i-\lfloor 2N/3\rfloor-1$ pulses in $[t'-T/2,\,t']$ and it did not reset its phase from $2\pi$ rad to $0$ within $(t'-T,\,t')$.	
			
	\item before receiving the current pulse, oscillator $i$ has already received at least $d_i-\lfloor 2N/3\rfloor-1$ pulses in $(t'-\epsilon,\,t']$.
    \end{enumerate}
	Otherwise, the pulse has no effect on $\phi_i$ who will evolve freely towards $2\pi$ rad.
\end{enumerate}
\rule{9cm}{0.12em}

\begin{Remark}\label{Remark_4_2_Global}
Following \cite{hong2005scalable,lucken2012two,nunez2015global, proskurnikov2016synchronization}, we assume that when a legitimate oscillator receives multiple pulses simultaneously, it can determine the number of received pulses and processes them consecutively. In other words, no two pulses will be regarded as an aggregated pulse.
\end{Remark}

\section{Synchronization of PCO Networks in the Presence Attacks}

In this section, we address the synchronization of PCO networks in the presence of Byzantine attacks. Among $N$ PCOs, we assume that $M$ are compromised and act as Byzantine attackers. We will show that Mechanism $1$ synchronizes legitimate oscillators even in the presence of multiple Byzantine attackers. Specifically, we will prove that under Mechanism $1$, legitimate oscillators achieve synchronization even when their topology is non-all-to-all and their initial phases are distributed arbitrarily in the entire oscillation period $[0,\,2\pi]$. More interestingly, when synchronization is achieved, the collective oscillation period of all legitimate oscillators is invariant under attacks and is identical to the free-running oscillation period $T=2\pi$ seconds. To facilitate theoretical analysis, we first establish Lemma \ref{Lemma_4_1_Global} about the properties of floor function $\lfloor\bullet\rfloor$.

\begin{Lemma}\label{Lemma_4_1_Global}
For three arbitrary positive integers $x$, $y$, and $Q$, with $x>y$, the following inequalities always hold:	
	\begin{equation*}
	\left\{
	\begin{array}{lr}
	\lfloor y\cdot Q/x \rfloor\geq y\cdot\lfloor Q/x \rfloor  \\
	\\
	\lfloor y\cdot Q/x \rfloor+\lfloor (x-y)\cdot Q/x  \rfloor+1\geq Q
	\end{array}
	\right.
	\end{equation*}	
\end{Lemma}
\begin{proof} First, we prove $\lfloor y\cdot Q/x \rfloor\geq y\cdot\lfloor Q/x \rfloor$. Since $x$ and $Q$ are positive integers, dividing $Q$ by $x$ and letting $q$ and $r$ be the quotient and remainder, respectively, we have $Q=x\cdot q+r$ and $0\leq r/x<1$. By substituting them into $\lfloor y\cdot Q/x \rfloor-y\cdot\lfloor Q/x \rfloor$, we have:
\begin{flalign}
\lfloor y\cdot Q/x \rfloor - y\cdot\lfloor Q/x \rfloor=&\lfloor y\cdot q+ {y\cdot r}/{x} \rfloor-y\cdot\lfloor q+ {r}/{x} \rfloor\nonumber\\
=& y\cdot q +\lfloor {y\cdot r}/{x} \rfloor-y\cdot  q \nonumber\\
=&\lfloor {y\cdot r}/{x} \rfloor\geq 0.\nonumber
\end{flalign}
Hence, we obtain $\lfloor y\cdot Q/x \rfloor\geq y\cdot\lfloor Q/x \rfloor$.

Next, we proceed to prove $\lfloor y\cdot Q/x \rfloor+\lfloor (x-y)\cdot Q/x  \rfloor+1\geq Q$. Dividing $y\cdot Q$ by $x$ and letting $\bar{q}$ and $\bar{r}$ be the quotient and remainder, respectively, we have $y\cdot Q=\bar{q}\cdot x+\bar{r}$ and $0\leq \bar{r}/x<1$. Substituting them into $\lfloor y\cdot Q/x \rfloor+\lfloor (x-y)\cdot Q/x  \rfloor+1-Q$ leads to
\begin{flalign}
&\lfloor y\cdot Q/x \rfloor+\lfloor (x-y)\cdot Q/x  \rfloor+1-Q\nonumber\\
=&\lfloor \bar{q}+\bar{r}/{x} \rfloor+\lfloor Q-\bar{q}-\bar{r}/{x} \rfloor+1-Q\nonumber\\
\geq&\lfloor\bar{q}\rfloor+\lfloor Q-\bar{q}-1 \rfloor+1-Q=0.\nonumber
\end{flalign}		
Thus, we obtain $\lfloor y\cdot Q/x \rfloor+\lfloor (x-y)\cdot Q/x  \rfloor+1\geq Q$.
\end{proof}

Now we are in position to prove that all legitimate oscillators will synchronize under Mechanism $1$ in the presence of Byzantine attacks even when legitimate oscillators are under a non-all-to-all connection and the initial phases are arbitrarily distributed in the entire oscillation period $[0,\,2\pi]$.

\begin{Theorem}\label{Theorem_4_1_Global}
For a network of $N$ PCOs among which $M$ are compromised and launch attacks following the Byzantine attack model in Sec \uppercase\expandafter{\romannumeral3}, if the degree of the PCO network satisfies $d>\lfloor 2N/3\rfloor$ and the number of attackers $M$ satisfies $M<d-\lfloor 2N/3\rfloor$, then all legitimate oscillators will synchronize under Mechanism $1$ from any initial phase distribution.
\end{Theorem}

\begin{proof} We set the initial time instant to $t_0$. The following proof is divided into two parts. In part \uppercase\expandafter{\romannumeral1}, we prove that all $N-M$ legitimate oscillators' phases reside in $[\pi,\,2\pi]$ at $t_0+T$ from any initial phase distribution. In Part \uppercase\expandafter{\romannumeral2}, we prove that these legitimate oscillators will reset their phases to $0$ at the same time and will keep having identical phases with a collective oscillation period $T=2\pi$ seconds, i.e., they will achieve synchronization.
	
\emph{Part \uppercase\expandafter{\romannumeral1} (all $N-M$ legitimate oscillators' phases reside in $[\pi,\,2\pi]$ at $t_0+T$)}: Since the number of attackers satisfies $M<d-\lfloor 2N/3\rfloor$ for $d\leq N-1$, using Lemma \ref{Lemma_4_1_Global}, we have
\[
M<d-\lfloor 2N/3\rfloor\leq N-1-\lfloor 2N/3\rfloor\leq\lfloor N/3\rfloor.
\]
According to the attacker model in Sec. \uppercase\expandafter{\romannumeral3}, we know that $M<\lfloor N/3\rfloor$ attackers can emit at most $M<\lfloor N/3\rfloor$ pulses within any time interval of length $\epsilon$. Since no legitimate oscillator fires within time interval $[t_0,\,t_0+T]$ under Mechanism $1$, a legitimate oscillator can receive at most $M<\lfloor N/3\rfloor$ pulses in any time interval of length $\epsilon$ within $[t_0,\,t_0+T]$. Therefore, upon reaching $2\pi$ rad within $[t_0,\,t_0+T]$, a legitimate oscillator will reset its phase to $\pi$ rad instead of $0$. 

Since the initial phases of all $N-M$ legitimate oscillators distribute arbitrarily in $[0,\,2\pi]$, at time $t_0$, they can be categorized into three possible scenarios, as depicted in Figure \ref{Theorem_1_Global_fig}:
\leftmargini=19mm
\begin{enumerate}
	\item[\emph{Scenario a)}:] all legitimate oscillators' initial phases reside in $[\pi,\,2\pi]$;
	\item[\emph{Scenario b)}:] all legitimate oscillators' initial phases reside in $[0,\,\pi)$;
	\item[\emph{Scenario c)}:] legitimate oscillators' initial phases reside partially in $[0,\,\pi)$ and partially in $[\pi,\,2\pi]$.
\end{enumerate}
\begin{figure}[htbp]
	\centering
	\includegraphics[width=0.35\textwidth]{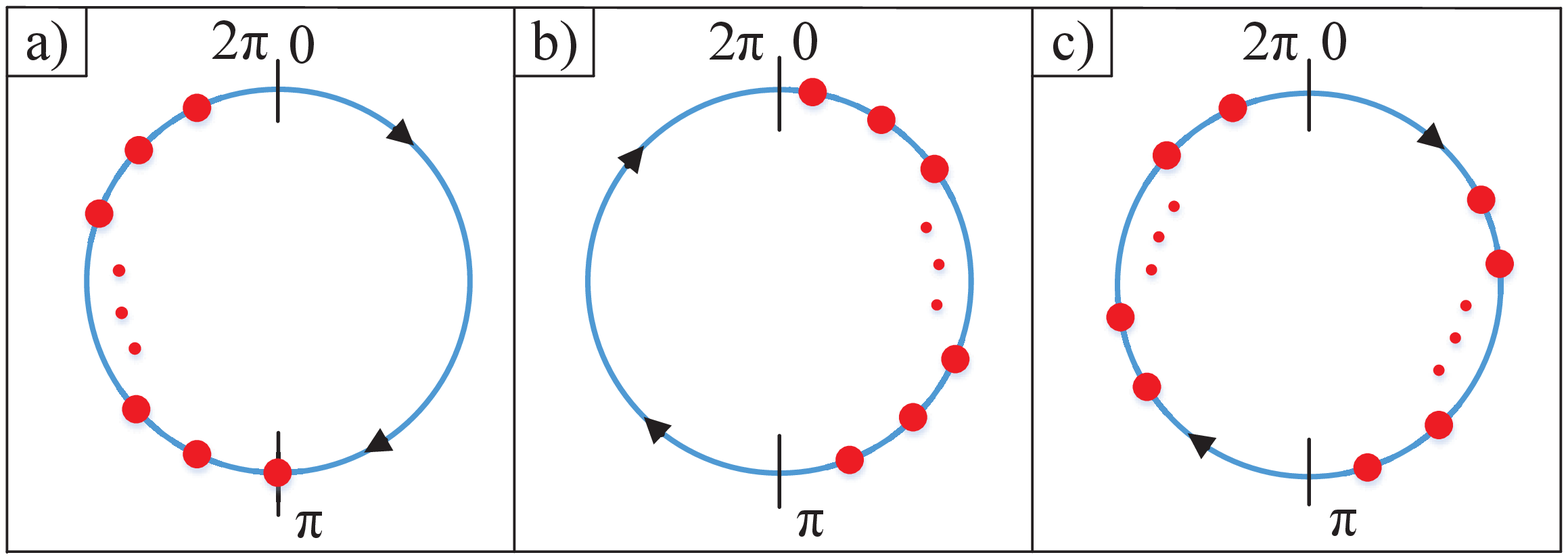}
	\caption{Three possible initial phase distributions of legitimate oscillators.}	
	\label{Theorem_1_Global_fig}
\end{figure}

Next, we show that no matter which of the three scenarios the initial phase distribution belongs to, all legitimate oscillators' phases will reside in $[\pi,\, 2\pi]$ at time $t_0+T$. We discuss all three scenarios of initial phase distribution one by one:

\emph{Scenario a)}: All legitimate oscillators' initial phases reside in $[\pi,\,2\pi]$. After reaching $2\pi$ rad within $[t_0,\,t_0+T]$, because a legitimate oscillator will receive less than $\lfloor N/3\rfloor$ pulses in any time interval of length $\epsilon$, it will reset its phase to $\pi$ rad according to Mechanism $1$. Therefore, we have that all legitimate oscillators will reside in $[\pi,\, 2\pi]$ at time $t_0+T$.
		
\emph{Scenario b)}: All legitimate oscillators' initial phases reside in $[0,\,\pi)$. According to Mechanism $1$, a legitimate oscillator will not respond to incoming pulses when its phase resides in $[0,\,\pi)$. So all legitimate oscillators' phases will evolve freely towards $\pi$ rad without perturbation and will enter $[\pi,\, 2\pi]$ no later than time instant $t_0+T/2$. After reaching $2\pi$ rad within $[t_0,\,t_0+T]$, because a legitimate oscillator will receive less than $\lfloor N/3\rfloor$ pulses in any time interval of length $\epsilon$, it  will reset its phase to $\pi$ rad according to Mechanism $1$. Therefore, we have that all legitimate oscillators' phases will reside in $[\pi,\, 2\pi]$ at time $t_0+T$.

\emph{Scenario c)}: Legitimate oscillators' initial phases reside partially in $[0,\,\pi)$ and partially in $[\pi,\,2\pi]$. Since legitimate oscillators with phases residing in $[0,\,\pi)$ will evolve freely into $[\pi,\,2\pi]$ under Mechanism $1$, we have that no later than time instant $t_0+T/2$, these oscillators' phase will be in $[\pi, 2\pi]$. Further making use of the fact that a legitimate oscillator will reset its phase to $\pi$ rad upon reaching $2\pi$ rad since less than $\lfloor N/3\rfloor$ pulses will be received by a single oscillator in any time interval of length $\epsilon$, we obtain that all legitimate oscillators' phases will reside in $[\pi,\, 2\pi]$ at time $t_0+T$.

Summarizing the above three scenarios, we have that regardless of the initial phase distribution, all legitimate oscillators' phases will reside in $[\pi,\, 2\pi]$ at time $t_0+T$ despite the presence of attacker pulses.

\emph{Part \uppercase\expandafter{\romannumeral2} (all legitimate oscillators will reset their phases to $0$ at the same time and will keep having identical phases with a collective oscillation period $T=2\pi$ seconds)}: From Part \uppercase\expandafter{\romannumeral1}, we know that no legitimate oscillator fires or resets its phase to $0$ within time interval $[t_0,\,t_0+T]$ and all legitimate oscillators' phases reside in $[\pi,\,2\pi]$ at time $t_0+T$. Therefore, all legitimate oscillators' phases will reach $2\pi$ rad and fire at least once within $(t_0+T,\,t_0+3T/2]$. Without loss of generality, we label all $N-M$ legitimate oscillators according to the order of their first firing time$\footnote{For example, if the firing sequence of legitimate oscillators A, B, C is A, A, B, A, C, then oscillators A,\,B,\,C are labeled as oscillators 1,\,2,\,3, respectively.}$ and denote $t_1\in(t_0+T,\,t_0+3T/2]$ as the first firing time of legitimate oscillator $\lfloor N/3\rfloor+1$. Only the following two scenarios could happen right before legitimate oscillator $\lfloor N/3\rfloor+1$ fires at $t_1$: 

\leftmargini=21.5mm
\begin{enumerate}
	\item[\emph{Scenario 1.1}:] no legitimate oscillator has reset its phase to $0$ before legitimate oscillator $\lfloor N/3\rfloor+1$ fires at $t_1$. 
	
	\item[\emph{Scenario 1.2}:] at least one legitimate oscillator has reset its phase to $0$ before legitimate oscillator $\lfloor N/3\rfloor+1$ fires at $t_1$. 
\end{enumerate}
Next, we show that in both scenarios all legitimate oscillators will reset their phases to $0$ at the same time and will keep having identical phases with a collective oscillation period $T=2\pi$ seconds, i.e., they will achieve synchronization. 

We first consider \emph{Scenario 1.1}, i.e., no legitimate oscillator has reset its phase to $0$ before legitimate oscillator $\lfloor N/3\rfloor+1$ fires at $t_1$. Since all the $N-M$ legitimate oscillators are labeled according to the order of their first firing time instants and no legitimate oscillator fired within $[t_0,\,t_0+T]$ according to Mechanism $1$, we have that before the firing of legitimate oscillator $\lfloor N/3\rfloor+1$ at $t_1$, $\lfloor N/3 \rfloor$ legitimate oscillators fired within time interval $(t_0+T,\,t_1]$ and every legitimate oscillator $i$ for $i=1,\,2,\,\cdots,N-M$ received at least $\lfloor N/3\rfloor-(N-d_i)$ pulses within time interval $(t_0+T,\,t_1]$, where $(N-d_i)$ is the number of oscillators which are not connected with oscillator $i$. According to Lemma \ref{Lemma_4_1_Global}, we have:
\begin{flalign}\label{Firing_number}
\lfloor N/3\rfloor-(N-d_i)=&\lfloor N/3\rfloor+\lfloor 2N/3\rfloor-N+d_i-\lfloor 2N/3\rfloor\nonumber\\
\geq&d_i-\lfloor 2N/3\rfloor-1
\end{flalign}
meaning that before the firing of legitimate oscillator $\lfloor N/3\rfloor+1$, every legitimate oscillator $i$ for $i=1,\,2,\,\cdots,N-M$ has already received at least $d_i-\lfloor 2N/3\rfloor-1$ pulses within time interval $(t_0+T,\,t_1]$ (note that this interval has length less than $T/2$).

When legitimate oscillator $i$ fires at $t_1$, at least $d$ legitimate oscillators will receive the pulse. As every legitimate oscillator has received at least $d_i-\lfloor 2N/3\rfloor-1$ pulses within $(t_0+T,\,t_1]$ (as proven in the previous paragraph), we have that for all legitimate oscillators, the condition $3a)$ of Mechanism $1$ is satisfied (note that in Scenario 1.1 we consider the case that no legitimate oscillators reset their phases to $0$ within $(t_1-T,\,t_1)$) and hence all legitimate oscillators that receive the pulse from legitimate oscillator $\lfloor N/3\rfloor+1$ (with quantity at least $d-M$) will shift their phases to $2\pi$ rad.


Next, we proceed to proved that among the $d-M$ legitimate oscillators whose phases are shifted to $2\pi$ rad by the pulse from legitimate oscillator $\lfloor N/3\rfloor+1$ at $t_1$, at least $d-M-\lfloor N/3\rfloor$ of them will fire. According to condition $2)$ of Mechanism $1$, if an oscillator fired within $(t_1-\epsilon,\,t_1]$, it cannot fire again at $t_1$. Since only $\lfloor N/3\rfloor$ legitimate oscillators fired before the firing of legitimate oscillator $\lfloor N/3\rfloor+1$ at $t_1$ (note that these oscillators might fire within $(t_1-\epsilon,\,t_1]$), we obtain that among the $d-M$ legitimate oscillators whose phases are shifted to $2\pi$ rad at $t_1$ by the pulse from legitimate oscillator $\lfloor N/3\rfloor+1$, at least $d-M-\lfloor N/3\rfloor$ of them will fire at $t_1$. From Lemma \ref{Lemma_4_1_Global} and making use of the fact $M<d-\lfloor 2N/3\rfloor$, we have
\[
 d-M-\lfloor N/3\rfloor>\lfloor N/3\rfloor
\]
meaning that the firing of legitimate oscillator $\lfloor N/3\rfloor+1$ will trigger at least $\lfloor N/3\rfloor+1$ other legitimate oscillators to fire simultaneously at $t_1$. The firing of these oscillators will further makes every legitimate oscillator $i$ for $i=1,2,\cdots,N-M$ to receive at least $d_i-\lfloor 2N/3\rfloor$ pulses at $t_1$ based on the relationship in (\ref{Firing_number}). Since all legitimate oscillators' phases reside in $[\pi,\,2\pi]$, according to Mechanism $1$, they will be shifted to $2\pi$ rad at $t_1$. Then, all the non-firing legitimate oscillators except those fired within the past $\epsilon$ time will fire at $t_1$.  

Recalling that only $\lfloor N/3\rfloor$ legitimate oscillators fired before legitimate oscillator $\lfloor N/3\rfloor+1$ fires at $t_1$, we obtain that at least $N-M-\lfloor N/3\rfloor$ legitimate oscillators will fire at $t_1$ and every legitimate oscillator $i$ for $i=1,2,\cdots,N-M$ will receive at least $N-M-\lfloor N/3\rfloor-(N-d_i)$ pulses from this firing event. According to Lemma \ref{Lemma_4_1_Global} and combining the fact $M<d-\lfloor 2N/3\rfloor$, we have 
\[
N-M-\lfloor N/3\rfloor-(N-d_i)=d_i-M-\lfloor N/3\rfloor>\lfloor N/3\rfloor
\]
meaning that every legitimate oscillator receives over $\lfloor N/3\rfloor$ pulses at $t_1$. Since every legitimate oscillator has phase residing on $2\pi$ and receives over $\lfloor N/3\rfloor$ pulses within $(t_1-\epsilon,\,t_1]$, all legitimate oscillators' phases will reset to $0$ after the firing event at $t_1$. 

Next, we proceed to prove that after time instant $t_1$, all legitimate oscillators will keep having identical phases and their collective oscillation period is $T=2\pi$ seconds, i.e., they will achieve synchronization.

From the above analysis, all legitimate oscillators' phases will be reset to $0$ at $t_1$. Because a legitimate oscillator's phase can only be affected by an incoming pulse when it resides in $[\pi,\,2\pi]$, we have that all legitimate oscillators' phases will evolve freely towards $\pi$ rad within time interval $(t_1,\,t_1+T/2)$. As soon as all legitimate oscillators' phases reach $\pi$ rad at time instant $t_1+T/2$, according to Mechanism $1$, legitimate oscillator $i$'s phase can be affected by an incoming pulse at time instant $t_1'\in[t_1+T/2,\,t_1+T)$ only if it receives over $d_i-\lfloor 2N/3\rfloor-1$ pulses within $(t_1'-\epsilon,\,t_1']$. Since the number of attackers satisfies $M\leq d-\lfloor 2N/3\rfloor-1\leq d_i-\lfloor 2N/3\rfloor-1$ and each attacker can emit at most one attack pulse within a time interval less than $\epsilon$, so attack pulses alone are not enough to trigger a phase shift on any legitimate oscillator's phase. Therefore, all legitimate oscillators will have identical phases and evolve freely towards $2\pi$ rad. 

At time instant $t_1+T$, all legitimate oscillators reach phase $2\pi$ rad and fire simultaneously, which makes legitimate oscillator $i$ for $i=1,\,2,\,\cdots, N-M$ receive at least $N-M-(N-d_i)=d_i-M>\lfloor N/3\rfloor$ pulses. Therefore, all legitimate oscillators will reset their phases to $0$ immediately. By repeating the above analyses, we can get that after time instant $t_1$, all legitimate oscillators will have identical phases with a collective oscillation period $T=2\pi$ seconds, i.e., phase synchronization of all legitimate oscillators is achieved immediately after time instant $t_1$.

Next, we consider \emph{Scenario 1.2}, i.e., at least one legitimate oscillator has reset its phase to $0$ before legitimate oscillator $\lfloor N/3\rfloor+1$ fires at $t_1$. Without loss of generality, we assume that legitimate oscillator $k$ is the first legitimate oscillator who resets its phase to $0$ within time interval $(t_0+T,\,t_1]$ and it resets its phase to $0$ at $t_k\in(t_0+T,\,t_1]$. According to Mechanism $1$, legitimate oscillator $k$ must have received over $\lfloor N/3\rfloor$ pulses within $(t_k-\epsilon,\,t_k]$.


We assume that legitimate oscillator $k$ receives the $\lfloor N/3\rfloor+1$'th pulse at time $t_k'$ within time interval $(t_k-\epsilon,\,t_k]$ and the pulse is sent by oscillator $k'$. According to condition $2)$ of Mechanism $1$, an oscillator can only fire once within $(t_k-\epsilon,\,t_k]$. So before the firing of oscillator $k'$ at $t_k'$, at least $\lfloor N/3\rfloor$ oscillators fired within $(t_k-\epsilon,\,t_k']$. Based on the relationship in (\ref{Firing_number}), every legitimate oscillator $i$ for $i=1,\,2,\,\cdots,N-M$ should have received at least $d_i-\lfloor 2N/3\rfloor-1$ pulses within $(t_k-\epsilon,\,t_k']$.

Then following the same line of reasoning in \emph{Scenario 1.1}, we have that the pulse of oscillator $k'$ will shift the phases of at least $d-M$ legitimate oscillators (which receive the pulse) to $2\pi$ rad and at least $\lfloor N/3\rfloor+1$ of them will fire at $t_k'$. Then, all legitimate oscillators' phases will be shifted to $2\pi$ rad and at least $N-M-\lfloor N/3\rfloor$ legitimate oscillators will fire at $t_k'$. Every legitimate oscillator will receive over $\lfloor N/3\rfloor$ pulses in this firing event at $t_k'$ and will reset its phase to $0$. We can also infer $t_k'=t_k=t_1$.

Next, following the same line of reasoning in \emph{Scenario 1.1}, we obtain that after time instant $t_1$, all legitimate oscillators will have identical phases and their collective oscillation period is $T=2\pi$ seconds, i.e., phase synchronization of all legitimate oscillators is achieved immediately after time instant $t_1$.
\end{proof}

\begin{Remark}
Theorem \ref{Theorem_4_1_Global} requires that the degree of the network is over $\lfloor 2N/3\rfloor$, which, according to \cite{wang_tsp:12}, also guarantees that the network is strongly connected.
\end{Remark}

\begin{Remark}
The mechanism requires that all legitimate oscillators to start at the same time instant. However, starting at the same time instant does not avoid dealing with arbitrary phase distribution since even after synchronization, for a non-all-to-all topology on which different attackers can affect different legitimate oscillators, attackers considered in this paper can disturb the phases of legitimate oscillators to an arbitrary distribution under existing pulse-coupled synchronization strategies.
\end{Remark}

\begin{Remark}
It is worth noting that the theoretical analysis in this paper is significantly different from our prior results in \cite{wang2018attack,wang2018pulse,wang2020attack}. In \cite{wang2018attack,wang2018pulse,wang2020attack}, we can prove that the length of the containing arc will decrease monotonically with time. However, in this paper, since the initial phases of all legitimate oscillators are arbitrarily distributed in the entire oscillation period and the considered attacker model is much stronger, such monotonic decreasing does not exist (see numerical simulation results in Figure \ref{Attack_N_known}, Figure \ref{Attack_N_unknown}, Figure \ref{Attack_N_known_Compare}, and Figure \ref{Attack_N_unknown_Compare}. Instead, we opt to prove that after initiation, our judiciously designed interaction mechanism can drive the phases of legitimate oscillators to within a half cycle in finite time. Then we proceed to prove that one legitimate oscillator's firing can (either directly or indirectly) trigger all legitimate oscillators to reset their phases to $0$ and the interaction mechanism can maintain phase synchronization even in the presence of attack pulses.
\end{Remark}

Mechanism $1$ can also guarantee synchronization of densely connected PCO networks in the absence of attacks, as detailed below:

\begin{Corollary}\label{Corollary_4_1_Global}
For a network of $N$ legitimate PCOs, if the degree of the PCO network satisfies $d>\lfloor 2N/3\rfloor$, then all oscillators will synchronize under Mechanism $1$ from an arbitrary initial phase distribution.
\end{Corollary}
\begin{proof} 
Corollary \ref{Corollary_4_1_Global} is a special case of Theorem \ref{Theorem_4_1_Global} when the number of attackers $M$ is set to $0$ and hence is omitted.
\end{proof}

\section{Extension to the Case where $N$ is Unknown to Individual Oscillators}

The implementation of Mechanism $1$ requires each node to have access to $N$, which may not be feasible in a completely decentralized network. Therefore, in this section, we propose a mechanism for the case where $N$ is unknown to individual oscillators. The essence is to leverage the degree information of individual oscillators, as detailed below:

\noindent\rule{9cm}{0.12em}
\emph{New Pulse-Based Synchronization Mechanism (Mechanism 2):}\\
\rule{9cm}{0.1em}
\begin{enumerate}
\item The phase $\phi_i$ of legitimate oscillator $i$ evolves from $0$ to $2\pi$ rad with a constant speed $\omega=1$ rad/second.
		
\item Once $\phi_i$ reaches $2\pi$ rad at time $t$, oscillator $i$ fires (emits a pulse) if it did not fire within $(t-\epsilon,\,t]$ and an entire period $T=2\pi$ seconds has elapsed since initiation. Then oscillator $i$ resets its phase from $2\pi$ rad to $0$ if it received at least $\lfloor d_i/3 \rfloor$ pulses within $(t-\epsilon,\,t]$. Otherwise, it resets its phase from $2\pi$ rad to $\pi$ rad.
		
\item When oscillator $i$ receives a pulse at time instant $t'$, it shifts its phase to $2\pi$ rad only if $\phi_i\in[\pi,\,2\pi]$ at time instant $t'$ and one of the following conditions is satisfied:
\begin{enumerate}		
\item before receiving the current pulse, oscillator $i$ has already received at least $\lfloor d_i/6\rfloor-1$ pulses in $[t'-T/2,\,t']$ and it did not reset its phase from $2\pi$ rad to $0$ within $(t'-T,\,t')$.	
			
\item before receiving the current pulse, oscillator $i$ has already received at least $\lfloor d_i/6\rfloor-1$ pulses in $(t'-\epsilon,\,t']$.
\end{enumerate}
Otherwise, the pulse has no effect on $\phi_i$ who will evolve freely towards $2\pi$ rad.
\end{enumerate}
\rule{9cm}{0.12em}

Following a similar line of reasoning in Section \uppercase\expandafter{\romannumeral5}, we can prove that Mechanism $2$ can synchronize densely connected PCO networks both in the presence and absence of Byzantine attackers.

\begin{Theorem}\label{Theorem_4_2_Global}
For a network of $N$ PCOs among which $M$ are compromised and launch attacks following the attack model in Sec \uppercase\expandafter{\romannumeral3}, if the degree of the PCO network satisfies $d>\lfloor3N/4\rfloor$ and the number of attackers $M$ satisfies $M<\lfloor d/6 \rfloor$, then all legitimate oscillators will synchronize under Mechanism $2$ from any initial phase distribution even if $N$ is unknown to individual oscillators.
\end{Theorem}

\begin{proof} We set the initial time instant to $t_0$. Similar to the proof in Theorem \ref{Theorem_4_1_Global}, the following proof is divided into two parts. In part \uppercase\expandafter{\romannumeral1}, we prove that all $N-M$ legitimate oscillators will have phases residing in $[\pi,\,2\pi]$ at $t_0+T$. In Part \uppercase\expandafter{\romannumeral2}, we prove that these legitimate oscillators will reset their phases to $0$ at the same time and will keep having identical phases with a collective oscillation period $T=2\pi$ seconds, i.e., they will achieve synchronization.

\emph{Part \uppercase\expandafter{\romannumeral1} (all $N-M$ legitimate oscillators' phases reside in $[\pi,\,2\pi]$ at $t_0+T$)}: Since the number of attackers satisfies $M<\lfloor d/6 \rfloor$, we have 
\[
M<\lfloor d/6 \rfloor\leq\lfloor d/3\rfloor\leq\lfloor d_i/3\rfloor
\]
for $i=1,\,2,\,\cdots,N-M$. Following the same line of reasoning in the proof of Theorem \ref{Theorem_4_1_Global}, Part \uppercase\expandafter{\romannumeral1}, we have that a legitimate oscillator will only reset its phases to $\pi$ rad within time interval $[t_0,\,t_0+T]$ and all legitimate oscillators' phases will reside in $[\pi,\,2\pi]$ at time instant $t_0+T$ no matter what the initial phase distribution is.

\emph{Part \uppercase\expandafter{\romannumeral2} (all legitimate oscillators will reset their phases to $0$ at the same time and will keep having identical phases with a collective oscillation period $T=2\pi$ seconds)}: Since no legitimate oscillator fires or resets its phase to $0$ within time interval $[t_0,\,t_0+T]$ and all legitimate oscillators' phases reside in $[\pi,\,2\pi]$ at time $t_0+T$, all legitimate oscillators' phases will reach $2\pi$ rad and fire at least once within $(t_0+T,\,t_0+3T/2]$. Without loss of generality, we label all $N-M$ legitimate oscillators according to the order of their first firing time and denote $t_1'\in(t_0+T,\,t_0+3T/2]$ as the first firing time instant of legitimate oscillator $\lfloor d/2\rfloor+1$. Only the following two scenarios could happen before legitimate oscillator $\lfloor d/2\rfloor+1$ fires at $t_1'$: 
\leftmargini=21.5mm
\begin{enumerate}
	\item[\emph{Scenario 2.1}:] no legitimate oscillator has reset its phase to $0$ before legitimate oscillator $\lfloor d/2\rfloor+1$ fires at $t_1'$. 
	
	\item[\emph{Scenario 2.2}:] at least one legitimate oscillator has reset its phase to $0$ before legitimate oscillator $\lfloor d/2\rfloor+1$ fires at $t_1'$. 
\end{enumerate}	
Next, we show that in both scenarios all legitimate oscillators will reset their phases to $0$ at the same time and will keep having identical phases with a collective oscillation period $T=2\pi$ seconds, i.e., they will achieve synchronization.

We first consider \emph{Scenario 2.1}, i.e., no legitimate oscillator has reset its phase to $0$ before legitimate oscillator $\lfloor d/2\rfloor+1$ fires at $t_1'$. Since all the $N-M$ legitimate oscillators are labeled according to the order of their first firing time instants and no legitimate oscillator fired within $[t_0,\,t_0+T]$ according to Mechanism $1$, we have that before the firing of legitimate oscillator $\lfloor d/2\rfloor+1$ at $t_1'$, $\lfloor d/2 \rfloor$ legitimate oscillators should have fired within time interval $(t_0+T,\,t_1']$ and every legitimate oscillator $i$ for $i=1,\,2,\,\cdots,N-M$ should have received at least $\lfloor d/2\rfloor-(N-d_i)$ pulses within $(t_0+T,\,t_1']$, where $(N-d_i)$ is the number of oscillators which are not connected to oscillator $i$. Since we have $d>\lfloor 3N/4\rfloor$, one can obtain $d_i\geq d\geq\lfloor 3N/4\rfloor+1>3N/4$ for $i=1,\,2,\,\cdots,N-M$. Using Lemma \ref{Lemma_4_1_Global} and combining the fact $d_i>3N/4$, we have:
\begin{flalign}\label{Firing_number_unknown}
\lfloor d/2\rfloor-(N-d_i)\geq&\lfloor d/2\rfloor-N+\lfloor 5d_i/6\rfloor+\lfloor d_i/6\rfloor\nonumber\\
\geq&\lfloor 3N/8\rfloor+\lfloor 5N/8\rfloor-N+\lfloor d_i/6\rfloor\nonumber\\
\geq&\lfloor d_i/6\rfloor-1
\end{flalign}
meaning that before the firing of legitimate oscillator $\lfloor d/2\rfloor+1$, every legitimate oscillator $i$ for $i=1,\,2,\,\cdots,N-M$ has already received at least $\lfloor d_i/6\rfloor-1$ pulses within time interval $(t_0+T,\,t_1']$ (note that this interval has length less than $T/2$). Then following the same line of reasoning in \emph{Scenario 1.1} of Theorem \ref{Theorem_4_1_Global}, we can prove that every legitimate oscillator $i$ will receive at least $\lfloor d_i/3 \rfloor$ pulses at $t_1'$ and reset its phases to $0$. Then starting from time instant $t_1'$, all legitimate oscillators will have identical phases with a collective oscillation period $T=2\pi$ seconds, i.e., they will achieve synchronization. 

The proof of \emph{Scenario 2.2} follows the same line of reasoning in \emph{Scenario 1.2} of Theorem \ref{Theorem_4_1_Global} and is omitted. 

Summarizing the above analyses, we conclude that Mechanism $2$ can synchronize densely connected PCO networks in the presence of Byzantine attacks even when $N$ is unknown to individual oscillators and initial phases are distributed arbitrarily.
\end{proof}

It is worth noting that Mechanism $2$ can also guarantee synchronization of densely connected PCO networks in the absence of attacks when $N$ is unknown to individual oscillators, as shown below:
 
\begin{Corollary}\label{Corollary_4_2_Global}
For a network of $N$ legitimate PCOs, if the degree of the PCO network satisfies $d>\lfloor 3N/4\rfloor$, then all oscillators will synchronize under Mechanism $2$ from any initial phase distribution even if $N$ is unknown to individual oscillators.
\end{Corollary}

\begin{proof} 
Corollary \ref{Corollary_4_2_Global} is a special case of Theorem \ref{Theorem_4_2_Global} when the number of attackers $M$ is set to $0$ and hence is omitted.
\end{proof}

\begin{Remark}\label{Remark_4_3_Global}
	According to Theorem \ref{Theorem_4_1_Global} and Theorem \ref{Theorem_4_2_Global}, Mechanism $1$ and Mechanism $2$ guarantee that all legitimate oscillators synchronize with a collective oscillation period $T=2\pi$ seconds (which is equal to the free-running period) even in the presence of Byzantine attacks. This is in distinct difference from existing results where the collective oscillation period is affected by attacks. 
\end{Remark}

\begin{Remark}
	When $N$ is unknown to individual oscillators, $d$ has to be larger than $\lfloor 3N/4\rfloor$, which is greater than $\lfloor 2N/3\rfloor$ for the case where $N$ is known. The requirement of increased connectivity is intuitive in that less knowledge of a PCO network requires stronger connectivity conditions to guarantee synchronization.
\end{Remark}


\section{Simulations}
We considered a network of $N=24$ PCOs placed on a circle with diameter $40$ meters as illustrated in Figure \ref{Topology}. Two oscillators can communicate if and only if their distance is less than $39$ meters. Thus, the degree of the network is $d=20$. We set $t_0=0$ and chose initial phases of oscillators randomly from $[0,\,2\pi]$. 

\begin{figure}[htbp]
	\centering
	\includegraphics[width=0.35\textwidth]{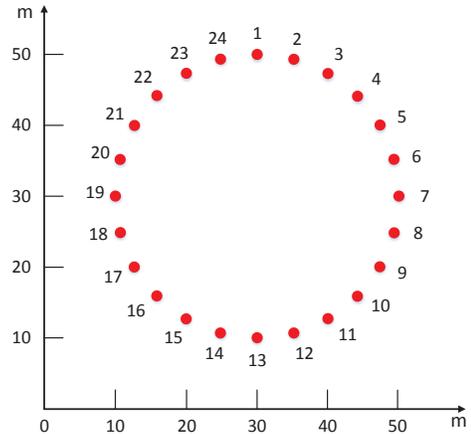}
	\caption{The deployment of the $24$ oscillators used in simulations.}
	\label{Topology}
\end{figure}

\subsection{In the Absence of Attacks}
We first considered the attacker-free case. As $d=20>\lfloor 3N/4\rfloor=18$, we know from Corollary \ref{Corollary_4_1_Global} and Corollary \ref{Corollary_4_2_Global} that the network will always synchronize from any initial phase distribution, whether or not $N$ is available to individual oscillators. This was confirmed by the numerical simulation results in Figure \ref{No_attack}. 

\begin{figure}[htbp]
	\centering
	\includegraphics[width=0.5\textwidth]{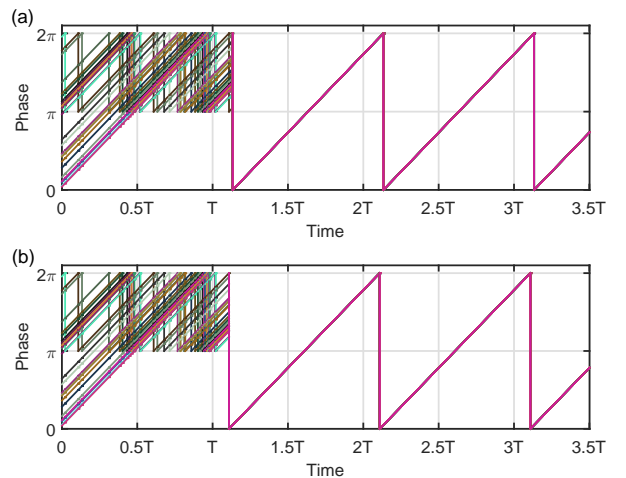}
	\caption{Plot (a) and (b) presented the phase evolutions of the $24$ PCOs under Mechanism $1$ and Mechanism $2$, respectively. $\epsilon$ was set to $0.01T$.}
	\label{No_attack}
\end{figure}

Using the same initial phase distribution as in Figure \ref{No_attack}, we also simulated the phase evolution of PCOs under the pulse-based synchronization mechanism in \cite{yun2015robustness}. It can be seen in Figure \ref{No_attack_optimal} that the pulse-based synchronization mechanism in \cite{yun2015robustness} cannot achieve synchronization, which shows the advantage of our new mechanisms even when attack-resilience is not relevant.

\begin{figure}[htbp]
	\centering
	\includegraphics[width=0.5\textwidth]{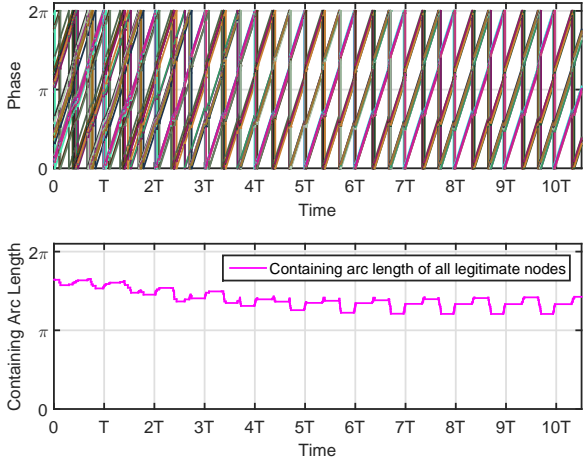}
	\caption{Phase evolution and the length of the containing arc of the $24$ PCOs under the pulse-based synchronization mechanism in \cite{yun2015robustness}. $l$ was set to $0.021$. }
	\label{No_attack_optimal}
\end{figure}

\subsection{In the Presence of Attacks}

Using the same network, we also ran simulations in the presence of Byzantine attacks when $N$ is known to individual oscillators.

We assumed that $3$ out of the $24$ PCOs (oscillators $1$, $8$, and $20$) were compromised and acted as Byzantine attackers. As $3< d-\lfloor 2N/3\rfloor=4$, we know from Theorem \ref{Theorem_4_1_Global} that the network will synchronize. This was confirmed by numerical simulations in Figure \ref{Attack_N_known}, which showed that even under Byzantine attacks the length of the containing arc of legitimate oscillators converged to zero. 

\begin{figure}[htbp]
	\centering
	\includegraphics[width=0.5\textwidth]{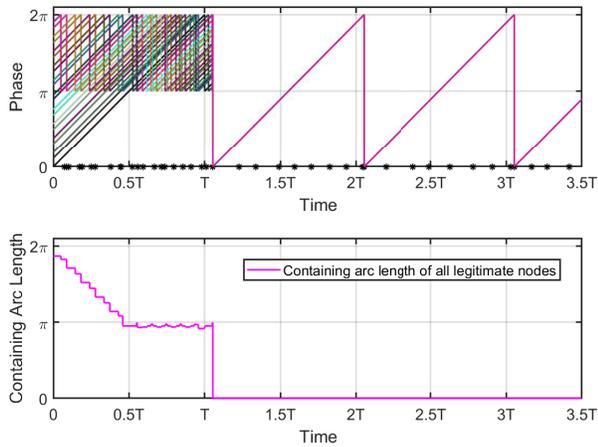}
	\caption{Phase evolution and the length of the containing arc of $21$ legitimate oscillators under Mechanism $1$ in the presence of $3$ Byzantine attackers (oscillators $1$, $8$, and $20$) with attacking pulse time instants represented by asterisks. $\epsilon$ was set to $0.01T$.}
	\label{Attack_N_known}
\end{figure}

Using the same network, when $N$ is unknown to individual oscillators, according to Theorem \ref{Theorem_4_2_Global}, the maximal allowable number of attackers is $\lfloor d/6 \rfloor-1=2$. Hence, the condition in Theorem \ref{Theorem_4_2_Global} was not satisfied. Simulation results confirmed that legitimate oscillators indeed could not synchronize as the collective oscillation period is time-varying and less than $T=2\pi$ seconds, which is illustrated in Figure \ref{Attack_N_unknown_Counter}.

\begin{figure}[htbp]
	\centering
	\includegraphics[width=0.5\textwidth]{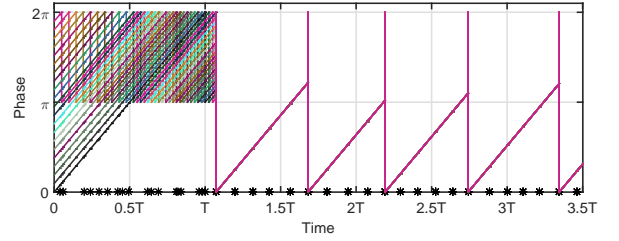}
	\caption{Phase evolution of $21$ legitimate oscillators under Mechanism $2$ in the presence of $3$ attackers (oscillators $1$, $8$, and $20$) with attacking pulse time instants represented by asterisks. $N$ was unknown to individual oscillators and $\epsilon$ was set to $0.01T$.}
	\label{Attack_N_unknown_Counter}
\end{figure}

However, when we decreased the number of attackers to $2$ (oscillators $1$ and $8$), all legitimate oscillators synchronized under Mechanism $2$ (see Figure \ref{Attack_N_unknown}), confirming the results in Theorem \ref{Theorem_4_2_Global}.

\begin{figure}[htbp]
	\centering
	\includegraphics[width=0.5\textwidth]{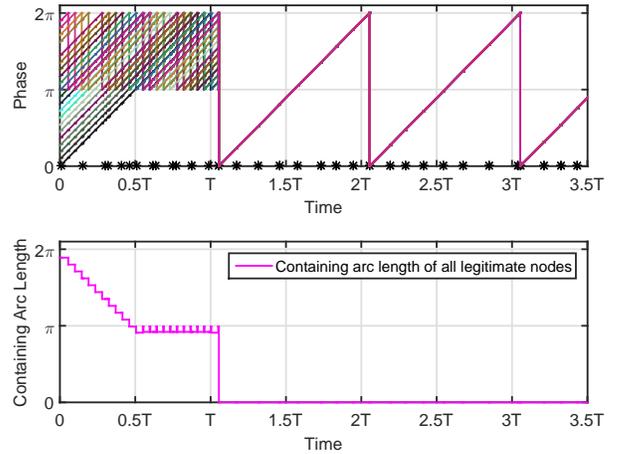}
	\caption{Phase evolution and the length of the containing arc of $22$ legitimate oscillators under Mechanism $2$ in the presence of $2$ attackers (oscillators $1$ and $8$) with attacking pulse time instants represented by asterisks. $N$ was unknown to individual oscillators and $\epsilon$ was set to $0.01T$.}
	\label{Attack_N_unknown}
\end{figure}

\subsection{Comparison with Existing Results}
Under the same PCO network deployment, we also compared our proposed Mechanisms $1$ and $2$ with existing attack resilient pulse-based synchronization approaches in \cite{yun2015robustness,wang2018attack,wang2018pulse,wang2020attack} which solely use content-free pulses in communications. When comparing with \cite{yun2015robustness,wang2018attack,wang2018pulse,wang2020attack}, we did not use the settings in \cite{yun2015robustness,wang2018attack,wang2018pulse,wang2020attack} since they are special cases of our setting, as can be seen in Table $1$.

Figure \ref{Attack_N_known_Compare} showed the evolutions of containing arc length of legitimate oscillators under Mechanism $1$ and approaches in \cite{yun2015robustness,wang2018attack,wang2018pulse,wang2020attack} in the presence of $3$ Byzantine attackers (oscillators $1$, $8$, and $20$) when $N$ was known to individual oscillators. All approaches used the same initial phase distribution (randomly chosen from $[0,\,2\pi]$) and identical malicious pulse attack patterns. It can be seen in Figure \ref{Attack_N_known_Compare} that Mechanism $1$ can achieve perfect synchronization whereas pulse-base synchronization approaches in \cite{yun2015robustness,wang2018attack,wang2018pulse,wang2020attack} failed to achieve synchronization even when the coupling strength was set to $l=1$. It is worth noting that similar results were obtained in all $1,000$ runs of our simulation with the initial phases randomly chosen from $[0,\,2\pi]$ and $40$ attack pulses randomly distributed in $[0,\,3.5T]$.

\begin{figure}[htbp]
	\centering
	\includegraphics[width=0.5\textwidth]{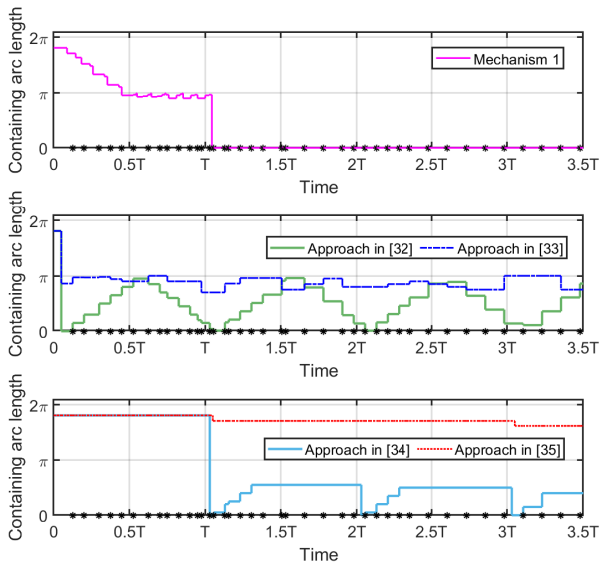}
	\caption{The length of the containing arc of $21$ legitimate oscillators under Mechanism $1$ and approaches in \cite{yun2015robustness,wang2018attack,wang2018pulse,wang2020attack} in the presence of $3$ Byzantine attackers (oscillators $1$, $8$, and $20$). The attack pulse time instants were represented by asterisks. The coupling strength in \cite{yun2015robustness,wang2018attack,wang2018pulse,wang2020attack} was set to $l=1$, $N$ was known to individual oscillators, and $\epsilon$ was set to $0.01T$.}
	\label{Attack_N_known_Compare}
\end{figure}

Figure \ref{Attack_N_unknown_Compare} showed the evolutions of containing arc length of legitimate oscillators under Mechanism $2$ and the approaches in \cite{yun2015robustness,wang2018attack,wang2018pulse,wang2020attack} in the presence of $2$ Byzantine attackers (oscillators $1$ and $8$) when $N$ was unknown to individual oscillators. Under the same set up, it can be seen in Figure \ref{Attack_N_unknown_Compare} that Mechanism $2$ can achieve perfect synchronization whereas existing pulse-base synchronization approaches in \cite{yun2015robustness,wang2018attack,wang2018pulse,wang2020attack} cannot, which confirmed the advantages of our new mechanism. It is worth noting that similar results were obtained in all $1,000$ runs of our simulation with the initial phases randomly chosen from $[0,\,2\pi]$ and $40$ attack pulses randomly distributed in $[0,\,3.5T]$.

\begin{figure}[htbp]
	\centering
	\includegraphics[width=0.5\textwidth]{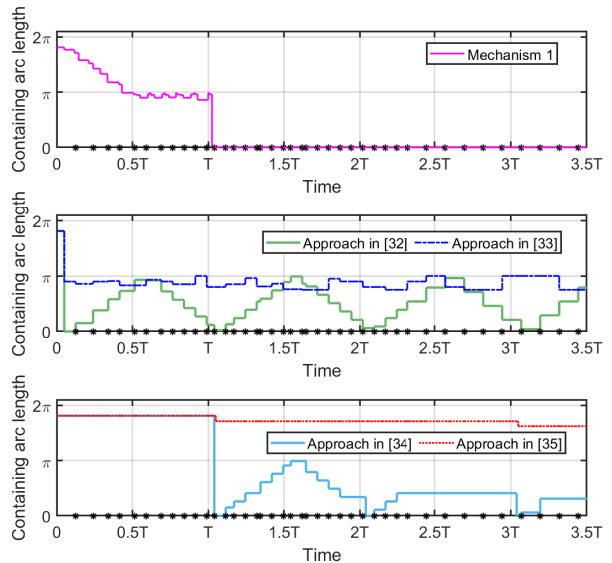}
	\caption{The length of the containing arc of $22$ legitimate oscillators under Mechanism $2$ and approaches in \cite{yun2015robustness,wang2018attack,wang2018pulse,wang2020attack} in the presence of $2$ Byzantine attackers (oscillators $1$ and $8$). The attack pulse time instants were represented by asterisks. The coupling strength in \cite{yun2015robustness,wang2018attack,wang2018pulse,wang2020attack} was set to $l=1$, $N$ was unknown to individual oscillators, and $\epsilon$ was set to $0.01T$.}
	\label{Attack_N_unknown_Compare}
\end{figure}

\section{Conclusions}

Due to unique advantages in simplicity, scalability, and energy efficiency over conventional packet-based synchronization approaches, pulse-based synchronization has been widely studied recently. However, all existing attack resilient pulse-based synchronization results are obtained either under all-to-all coupling topology or restricted initial phase distributions. In this paper, we propose a new pulse-based interaction mechanism to improve the resilience of PCO networks against Byzantine attackers. The new mechanism can enable synchronization in the presence of multiple Byzantine attackers even when the PCO network is not restricted to all-to-all and the initial phases are distributed arbitrarily. This is in distinct difference from most of the existing attack resilience algorithms which require a priori (almost) synchronization among all legitimate oscillators. The approach is also applicable when the total number of oscillators is unknown to individual oscillators. Numerical simulations confirmed the analytical results. In future work, we plan to relax the condition that all legitimate oscillators start at the same time instant and allow different oscillators to be turned on at different time instants.

\bibliographystyle{unsrt}
\bibliography{abbr_bibli}
\end{document}